\documentclass[runningheads, envcountsame, a4paper]{llncs}

\usepackage[pdftex]{graphicx}
\usepackage{epstopdf}

\usepackage[hyphens]{url}

\usepackage[misc]{ifsym}

\usepackage{amsmath}

\usepackage[ruled,linesnumbered]{algorithm2e}
\usepackage{bm}

\usepackage{amsfonts}
\usepackage{amssymb}

\usepackage{microtype}

\DeclareMathOperator*{\argmax}{arg\,max}

\newcommand{\Osymbol}{{\mathcal O}}
\newcommand{\BO}[1]{\Osymbol\left(#1\right)}
\newcommand{\TilO}[1]{\tilde{\Osymbol}\left(#1\right)}

\newcommand{\E}[1]{\textrm{\bf E}\left[#1\right]}
\renewcommand{\Pr}[1]{\textrm{\bf Pr}\left[#1\right]}

\newcommand{\sgn}{\texttt{sgn}}

\newcommand{\Rd} {\mathbb{R}^d}

\newcommand{\Rnd} {\mathbb{R}^{n \times d}}

\newcommand{\mX} {\mathbb{X}}

\newcommand{\bX} {\mathbf{X}}
\newcommand{\bQ} {\mathbf{Q}}

\newcommand{\bp} {\mathbf p}
\newcommand{\bq} {\mathbf q}

\newcommand{\bx} {\mathbf x}
\newcommand{\by} {\mathbf y}

\begin{document}

\title{Revisiting Wedge Sampling for Budgeted Maximum Inner Product Search \thanks{Partially supported by the Innovation Fund Denmark through the DABAI project.}
}

\toctitle{Revisiting Wedge Sampling for Budgeted Maximum Inner Product Search}
	
\author{Stephan S. Lorenzen\inst{1} \and
	Ninh Pham \Letter \, \inst{2}\orcidID{0000-0001-5768-9900} }

\authorrunning{S. Lorenzen and N. Pham}
\tocauthor{Stephan S. Lorenzen and Ninh Pham}
%
\institute{
	University of Copenhagen, Copenhagen, Denmark \and
	University of Auckland, Auckland, New Zealand\\
	\email{lorenzen@di.ku.dk, ninh.pham@auckland.ac.nz}
}

%
%

\maketitle              

\setcounter{footnote}{0}
	
\begin{abstract}

Top-$k$ maximum inner product search (MIPS) is a central task in many machine learning applications. 
This work extends top-$k$ MIPS with a budgeted setting, that asks for the best approximate top-$k$ MIPS given a limited budget of computational operations.
We investigate recent advanced sampling algorithms, including wedge and diamond sampling, to solve budgeted top-$k$ MIPS.
First, we show that diamond sampling is essentially a combination of wedge sampling and basic sampling for top-$k$ MIPS.
Our theoretical analysis and empirical evaluation show that wedge is competitive (often superior) to diamond on approximating top-$k$ MIPS regarding both efficiency and accuracy.
Second, we propose \textit{dWedge}, a very simple \emph{deterministic} variant of wedge sampling for budgeted top-$k$ MIPS.
Empirically, dWedge provides significantly higher accuracy than other budgeted top-$k$ MIPS solvers while maintaining a similar speedup. 

\keywords{Budgeted Maximum Inner Product Search \and	Sampling}
\end{abstract}

\section{Introduction}

Maximum inner product search (MIPS) is the task of, given a point set $\mX \subset \Rd$ of size $n$ and a query point $\bq \in \Rd$, finding the point $\bp \in \mX$ such that, $$\bp = \argmax_{\bx \in \mX}{\bx \cdot \bq} \enspace .$$
MIPS and its variant top-$k$ MIPS, which finds the top-$k$ largest inner product points with a query, are central tasks in the retrieval phase of standard collaborative filtering based recommender systems~\cite{Cremonesi10,KorenIEEE09}.
They are also algorithmic ingredients in a variety of machine learning tasks, for instance, prediction tasks on multi-class learning~\cite{Dean13,ImageNet} and neural network~\cite{Covington16,Spring17}.

Modern real-world online recommender systems often deal with very large-scale data sets and a limited amount of response time. 
Such collaborative filtering based systems often present users and items as low-dimensional vectors.
A large inner product between these vectors indicates that the items are relevant to the user preferences. 
The recommendation is often performed in the \textit{online} manner since the user vector is updated online with ad-hoc contextual information only available during the interaction~\cite{Xbox,YahooMusic,Koren09}.
A personalized recommender needs to infer user preferences based on online user behavior, e.g. recent search queries and browsing history, as implicit feedback to return relevant results~\cite{Hu08,Rendle09}.
Since the retrieval of recommended items is only performed online, the result of this task might not be ``perfect'' given a small amount of waiting time but its accuracy/relevance should be improved given more waiting time.
Hence, it is challenging to not only speed up the MIPS process, but to trade the search efficiency for the search quality for performance improvement.

Motivated by the computational bottleneck in the retrieval phase of modern recommendation systems, this work investigates the \emph{budgeted} MIPS problem, a natural extension of MIPS  with a computational limit for the search efficiency and quality trade-off.
Our budgeted MIPS addresses the following question:

\medskip

{\em Given a data structure built in $\TilO{dn}$ time~\footnote{Polylogarithmic factors, e.g. $\log{d}\log{n}$ is absorbed in the $\tilde{\Osymbol}$-notation.} and budgeted computational operations, can we have an algorithm to return the best approximate top-$k$ MIPS?}

\medskip

To measure the accuracy of approximate top-$k$ MIPS, we use the search \emph{recall}, i.e. the empirical probability of retrieving the true top-$k$ MIPS.
In our budgeted setting, we limit the time complexity of building a data structure to $\TilO{dn}$ since when a context is used in a recommender system, the learning phase cannot be done entirely offline~\cite{Xbox,YahooMusic}. In other words, the item vectors are also computed online and hence a high cost of constructing the data structure will degrade the performance.
Furthermore, since user preferences often change over time, a recommender system needs to frequently update its factorization model to address such \emph{drifting} user preferences.
This means that the item vectors and our data structure will be updated frequently.
 
It is worth noting that the budgeted MIPS has been recently studied in~\cite{Yu17} given a budget of $B = \eta n$ inner product computation where $\eta$ is a small constant, e.g. 5\%.
Furthermore, such budget constraints on the number of computational operations or on accessing a limit number of data points are widely studied not only on search problems~\cite{Ram12} but also on clustering~\cite{Mai13,Shamir11} and other problems~\cite{Fetaya15,Zilberstein96} when dealing with large-scale complex data sets.

\subsection{Prior art on solving MIPS and its limit on budgeted MIPS}

It is well-known that due to the ``curse of dimensionality'', any exact solution for MIPS based on data or space partitioning indexing data structures generally degrades when dimensionality increases.
It is no better than a simple sequential scanning when the dimensionality is larger than~10~\cite{FEXIPRO,Weber98}. 
Hence recent work on solving MIPS focuses on speeding up sequential scanning by pruning the search space~\cite{Abuzaid19,FEXIPRO,LEMP}.
Though such methods can solve MIPS exactly, they require $\Theta(n)$ operations and therefore do not fit well to the budgeted MIPS setting where we might need $o(n)$ operations for some query.
Furthermore, these methods do not provide any trade-off between the search quality and efficiency for online queries.

Another research direction is investigating approximation solutions that trade accuracy for efficiency.
Since \emph{locality-sensitive hashing} (LSH)~\cite{Har12} has emerged as a basic algorithmic tool for similarity search in high dimensions due to the sublinear query time guarantee, several approaches have followed this direction to obtain sublinear solutions for approximate MIPS~\cite{Huang18,Neyshabur15,Shrivastava14,Yan18}.
Due to the inner product not being a metric, these LSH-based solutions have to convert MIPS to near neighbor search problem by applying order-preserving transformations to exploit the LSH framework.

Although LSH-based approaches can guarantee sublinear query time, the top-$k$ inner product values are often very small compared to the vector norms in high dimensions.
This means that the distance gap between ``close'' and ``far apart'' points after transformation in the LSH framework is arbitrarily small.
That leads to not only the space usage (i.e. the number of hash tables) blow up, but also degrading LSH performance~\cite{Ahle16}.
Furthermore, the LSH trade-off between search quality and efficiency is somewhat ``fixed'' for any query since it is governed by specific parameters of the LSH data structure, e.g. number of hash functions.
We will show these limitations in our experiment section.
Note that the learning phase of a recommender system has to be executed in an online manner when a context is used~\cite{Xbox,YahooMusic}.
In this case, the significant cost of building LSH tables will be a computational bottleneck for handling online recommendations.

An alternative efficient solution is applying sampling methods to estimate the matrix-vector multiplication derived by top-$k$ MIPS~\cite{Diamond,Wedge}.
The basic idea is to sample a point $\bx$ with probability proportional to the inner product $\bx \cdot \bq$.
The larger inner product values the point $\bx$ has, the more occurrences of $\bx$ in the sample set.
By the end of the sampling process, we extract top-$B$ points ($B > k$) with the largest occurrences in the sample set via a counting histogram.
The top-$k$ points with the largest inner product among these $B$ points will be returned as an approximate top-$k$ MIPS.
Such sampling schemes naturally fit the budgeted setting since the more samples provide the higher accuracy.

\subsection{Our contribution}

This work studies sampling methods for solving the budgeted MIPS since they naturally fit to the class of budgeted problems.
Sampling schemes provide not only the trade-off between search quality and search efficiency but also a flexible mechanism to control this trade-off via the number of samples $S$ and the number of inner product computation $B$.
Our contributions are as follows:
\begin{enumerate}
	\item We revise popular sampling methods for solving MIPS, including basic sampling, wedge sampling~\cite{Wedge}, and the state-of-the-art diamond sampling methods~\cite{Diamond}.
	We show that diamond sampling is essentially a combination of basic sampling and wedge sampling.
	\item Our novel theoretical analysis and empirical evaluation illustrate that 
	wedge is competitive (often superior) to diamond on approximating top-$k$ MIPS regarding both efficiency and accuracy.
	%
	\item We propose \emph{dWedge}, a very simple but efficient \emph{deterministic} variant of wedge sampling, with a flexible mechanism to govern the trade-off between search quality and efficiency for the budgeted top-$k$ MIPS.
	
	\item Empirically, dWedge outperforms other competitive budgeted MIPS solvers~\cite{Neyshabur15,Yan18,Yu17} on standard recommender system data sets.
	Especially, dWedge returns the top-10 MIPS with at least 90\% accuracy with the speedup between 20x and 180x compared to the brute-force search on our large-scale data sets.	
	%

\end{enumerate}

\section{Notation and Preliminaries} \label{sec:background}

We present the point set $\mX$ as a matrix ${\bX} \subset \Rnd$ where each point $\bx_i$ corresponds to the $i$th row, and the query point $\bq$ as a column vector $\bq = (q_{1}, \ldots, q_{d})^T$.
We use $i \in [n]$ to index row vectors of ${\bX}$, i.e. $\bx_i = (x_{i1}, \ldots, x_{id}) \in \Rd$.
Since we will describe our investigated methods using the column-wise matrix-vector multiplication $\bX \bq$, we use $j \in [d]$ to index column vectors of $\bX$, i.e. $\by_j = (x_{1j}, \ldots, x_{nj})^T \in \mathbb{R}^n$.
For each column $j$, we pre-compute its 1-norm $c_j = \|y_j\|_1$.

We briefly review sampling approaches for estimating inner products $\bx_i \cdot \bq$.
For simplicity, we first assume that $\bX$ and $\bq$ are non-negative.
Then we show how to extend these approaches to handle negative inputs with their limits.
We consider the column-wise matrix-vector multiplication $\bX \bq$ as follows.
\begin{equation}
\begin{aligned}
\label{eq:Xq}
{\bX \bq} &= 
\begin{bmatrix}
x_{11} \\
\vdots \\
x_{n1}
\end{bmatrix} q_1 
+ 
\begin{bmatrix}
x_{12} \\
\vdots \\
x_{n2}
\end{bmatrix}  q_2
+ \ldots
+ \begin{bmatrix}
x_{1d} \\
\vdots \\
x_{nd}
\end{bmatrix} q_d \\
&= \by_1 q_1 + \by_2 q_2 + \ldots + \by_d q_d  
\end{aligned}
\end{equation}


\subsection{Basic Sampling}\label{sec:random}

Basic sampling is a very straightforward method to estimate the inner product $\bx_i \cdot \bq$ for the point $\bx_i$.
For any row $i$, we sample a column $j$ with probability $q_j/ \| \bq \|_1$ and return $x_{ij}$.
Define a random variable $Z_i = x_{ij}$, we have
\begin{align*}
\E{Z_i} = \sum_{j = 1}^{d} x_{ij} \frac{q_j}{\|\bq\|_1} = \frac{\bx_i \cdot \bq}{\|\bq\|_1} \ .
\end{align*}

The basic sampling suffers large variance when most of the contribution of $\bx_i \cdot \bq$ are from a few coordinates.
In particular, the variance will be significantly large when the main contributions of $\bx_i \cdot \bq$ are from a few coordinates $x_{ij}q_j$ and $q_j$ are very small.
Note that this basic sampling approach has been used in~\cite{Drineas06} as an efficient sampling technique for approximating matrix-matrix multiplication.

\textbf{Negative inputs:} To handle the negative cases, one can change the sampling probability to  $|q_j|/ \| \bq \|_1$ and return $Z_i = \sgn(q_j)x_{ij}$ where $\sgn$ is the sign function, i.e. $\sgn(u) = -1$ if $u < 0$ and $\sgn(u) = 1$ if $u \geq 0$.
It is clear that $\E{Z_i} = \bx_i \cdot \bq / \|\bq\|_1$.
Despite providing the unbiased estimate, this scheme needs $S = \Omega(n)$ samples for estimating $n$ inner product values to answer top-$k$ MIPS.

\subsection{Wedge Sampling}\label{sec:wedge}


\begin{algorithm}[!t]
\SetAlgoLined
		\KwData { Matrices $\bX$, query $\bq$, pre-computed values $z$ and $c_j$ for each $j \in [d]$, number of samples $S$ and number of inner products $B$. }
		\KwResult { Approximate top-$k$ MIPS for $\bq$.  }
		\textbf{Screening:} Wedge sample $S$ points and increase its counter value.\\
		Extract top-$B$ points with the largest values from the counter histogram.\\
		\textbf{Ranking}: Compute these $B$ inner products and return top-$k$ points with the largest inner product values.
		%
		\normalsize
	\caption{ Wedge sampling}
	\label{alg:wedge} 
\end{algorithm}

Cohen and Lewis~\cite{Wedge} proposed an efficient sampling approach, called wedge sampling, to approximate matrix multiplication and to isolate the largest inner products as a byproduct.
Wedge sampling needs to pre-compute some statistics, including the sum of all inner products $z = \sum_{i}z_i$ where $z_i = \bx_i \cdot \bq$ and 1-norm of column vectors $c_j = \| \by_j \|_1$.
Since we can pre-compute $c_j$ \emph{before} querying, computing $z = \sum_{j} q_j c_j$ clearly takes $\BO d$ query time.
We can think of $q_jc_j/z$ as the contribution ratio of the column $j$ to the sum of inner product values $z$.

The basic idea of wedge sampling is to randomly sample a row index $i$ corresponding to $\bx_i$ with probability $z_i / z$.
Hence, the larger the inner product $z_i = \bx_i \cdot \bq$, the larger the number of occurrences of $i$ in the sample set.
Consider Equation~(\ref{eq:Xq}), wedge sampling first samples a column $j$ corresponding to $\by_j$ with probability $q_jc_j/z$, and then samples a row $i$ corresponding to $\bx_i$ from $\by_j$ with probability $x_{ij}/c_j$.
By Bayes's theorem, we have
\begin{align*}
\Pr{\texttt{Sampling } i} &= 
\sum_{j=1}^{d}{\Pr{\texttt{Sampling } i | \texttt{Sampling } j}} \cdot \Pr{\texttt{Sampling } j} \\
&= \sum_{j=1}^{d} \frac{x_{ij}}{c_j} \cdot  \frac{q_jc_j}{z} = \frac{\sum_{j=1}^{d}{x_{ij} q_j}}{z} = \frac{z_i}{z} \enspace .
\end{align*}

Applying wedge sampling method on $\bX \bq$, we obtain a sample set where each index $i$ corresponding to $\bx_i$ is sampled according to an independent Bernoulli distribution with parameter $p_i = z_i / z$.
As a screening phase, a simple counting algorithm will be used to find the points with the largest counters.
Given $S$ samples and a constant cost for each sample, such a counting algorithm runs in $\BO{S + min(S, n)\log{k}}$ time to answer approximate top-$k$ MIPS.
If we have an additional budget of $B > k$ inner product computation, we can compute the exact inner product values of the top-$B$ points with the largest counter values for ranking. 
Such ranking (or post-processing) phase with an additional $\BO{dB}$ (since $d > \log{k}$) computational cost will provide higher accuracy for top-$k$ MIPS in practice.
Algorithm~\ref{alg:wedge} shows how the wedge sampling works.

We note that since wedge sampling uses the contribution ratio $q_jc_j/z$ to sample the column $j$, it can alleviate the effect of skewness of $\bx_i \cdot \bq$ where large contributions are from a few coordinates.
Hence wedge sampling achieves lower variance than the basic sampling in practice.

\textbf{Negative inputs:} Again, we can use the sign trick to deal with negative cases.
We note that this trick has been first exploited in the diamond sampling approach~\cite{Diamond}.
In particular, we execute wedge sampling on absolute values of $\bX$ and $\bq$, and return $Z_i = \sgn(x_{ij})\sgn(q_{j})$ for the point $\bx_i$.
It is clear that $\E{Z_i}$ is proportional to $\bx_i \cdot \bq$.
The analysis of this trick under some assumptions of the data distribution can be found in~\cite{Ding19}.

\subsection{Diamond Sampling}\label{sec:diamond}

Ballard et al.~\cite{Diamond} proposed diamond sampling to find the largest \emph{magnitude} elements from a matrix-matrix multiplication $\bX \bQ$ without computing the final matrix directly.
The method first presents $\bX \bQ$ as a weighted tripartite graph. 
Then it samples a diamond, i.e. four cycles from such graph with probability proportional to the value $(\bX \bQ)^2_{ij}$, which claims to amplify the focus on the largest magnitude elements.

Consider a vector $\bq$ as a one-column matrix $\bQ$, it is clear that diamond sampling can be applied to solve MIPS.
Indeed, we will show that diamond sampling is essentially a combination of wedge sampling and basic sampling when approximating $\bX \bq$.
In particular, diamond sampling first makes use of wedge sampling to return a random row $i$ corresponding to $\bx_i$ with probability $z_i/z$.
Given such row $i$, it then applies basic sampling to sample a random column $j'$ with probability $q_{j'}/ \| \bq \|_1$ and return $x_{ij'}$ as a scaled estimate of $\left(\bx_i \cdot \bq\right)^2$.
Define a random variable $Z_i = x_{ij'}$ corresponding to $\bx_i$, using the properties of wedge sampling and basic sampling we have
\begin{align*}
\E{Z_i} = \sum_{j' = 1}^{d} x_{ij'} \frac{q_{j'}}{\|\bq\|_1} \cdot \frac{z_i}{z} = \frac{(\bx_i \cdot \bq)^2}{z\|\bq\|_1} \enspace .
\end{align*}

Since diamond sampling builds on basic sampling, it
suffers from the same drawback as basic sampling.
To answer top-$k$ MIPS, diamond follows the same procedure as wedge hence shares the same asymptotic running time.


\textbf{Negative inputs:} Handling negative cases using diamond is similar to wedge.
We apply diamond sampling on absolute values of $\bX$ and $\bq$ then return $Z_i = \sgn(q_{j}) \sgn(x_{ij}) \sgn(q_{j'})x_{ij'}$ where $j$ is the column sampled by wedge sampling and $j'$  is the column sampled by basic sampling.
We can verify that $\E{Z_i}$ is proportional to $\left(\bx_i \cdot \bq \right)^2$.

Although diamond sampling can deal with negative inputs, its concentration bound only works on non-negative cases.
Furthermore, diamond sampling indeed solves a different problem, i.e. $\argmax_i{(\bx_i \cdot \bq)^2}$, which will give a completely different result on negative inputs.
In practice, the implementation of diamond sampling requires significant query time overhead due to the basic sampling.
This sampling process generates random variables from a discrete distribution derived from the query $\bq$ and requires expensive \textit{random} access operations to access $x_{ij'}$.
We will show that wedge significantly outperforms diamond regarding both accuracy and efficiency on our benchmark data sets.

\section{Wedge Sampling for Budgeted Top-$k$ MIPS}

This section first presents a new analysis of wedge sampling.
Our theoretical concentration bound shows that wedge requires fewer samples than diamond on approximating top-$k$ MIPS.
Then we present a drawback of wedge sampling for the budgeted MIPS and propose \emph{dWedge}, a simple deterministic variant to handle such drawback.
dWedge can govern the trade-off between search quality and efficiency with two parameters: the number of samples $S$ and the number of $B$ inner product computation.

\subsection{Analysis of Wedge Sampling}

This subsection shows the analysis of wedge sampling on non-negative inputs.
Consider a counting histogram of $n$ counters corresponding to $n$ point indexes, the following theorem states the number of samples required to distinguish between two inner product values $\tau_1$ and $\tau_2$.

\begin{theorem}\label{thm:main}
Fix two thresholds $\tau_1 > \tau_2 > 0$ and suppose $S \geq \frac{3z\ln{n}}{(\sqrt{\tau_1} - \sqrt{\tau_2})^2}$ where $z = \sum_{i}{\bx_i \cdot \bq}$. 
With probability at least $1 - \frac{1}{n}$, the following holds for all pairs $i_1, i_2 \in [n]$: if $\bx_{i_1} \cdot \bq \geq \tau_1$ and $\bx_{i_2} \cdot \bq \leq \tau_2$, then $counter[i_1] > counter[i_2]$.
\end{theorem}
\begin{proof}
Define $p_1 = \frac{\bx_{i_1} \cdot \bq}{z} \geq \frac{\tau_1}{z}$ and $p_2 = \frac{\bx_{i_2} \cdot \bq}{z} \leq \frac{\tau_2}{z}$.
For $l = 1, \ldots, S$, we consider independent pair of random variables $(X_l, Y_l)$ where
\begin{align*}
X_l = 
\left\{ 
				\begin{array}{rl}
 								1 & \mbox{if } \bx_{i_1} \mbox{is chosen at $l$th sample;} \\
 								0 & \mbox{otherwise}
        \end{array} 
\right. 
\enspace \enspace
Y_l = 
\left\{ 
				\begin{array}{rl}
 								1 & \mbox{if } \bx_{i_2} \mbox{is chosen at $l$th sample;} \\
 								0 & \mbox{otherwise.}
        \end{array} 
\right.
\end{align*}
Define $X = \sum_{l = 1}^{S}X_l$ and $Y = \sum_{l = 1}^{S}Y_l$.
We only consider the failure case where $ Y - X \geq 0$.
Applying Markov inequality for any $\lambda > 0$, we have
\begin{align*}
\Pr{Y - X \geq 0} &= \Pr{e^{\lambda(Y - X)} \geq 1} \leq \E{e^{\lambda(Y - X)}}  \\
&= \E{e^{\lambda \left(\sum_l{Y_l} - \sum_l{X_l} \right)}} = \prod_{l = 1}^{S}{\E{e^{\lambda(Y_l - X_l)}}}\enspace.
\end{align*}
We also have
\begin{align*}
\E{e^{\lambda(Y_l - X_l)}} &= e^\lambda p_2 + \left(1 - p_1 - p_2 \right) + e^{-\lambda} p_1 \\
&\geq 2\sqrt{p_1 p_2} + 1 - p_1 - p_2 = 1 - \left(\sqrt{p_1} - \sqrt{p_2} \right)^2 \enspace.
\end{align*}
The equality holds when $\lambda = \ln{\sqrt{p_1 / p_2}} > 0$.
In other words, by choosing $\lambda = \ln{\sqrt{p_1 / p_2}}$, we have
\begin{align*}
\Pr{Y - X \geq 0} &\leq \left( 1 - \left(\sqrt{p_1} - \sqrt{p_2} \right)^2  \right)^S \leq e^{-S\left( \sqrt{p_1} - \sqrt{p_2} \right)^2} \enspace .
\end{align*}
By choosing $S \geq \frac{3z\ln{n}}{(\sqrt{\tau_1} - \sqrt{\tau_2})^2} \geq \frac{3\ln{n}}{\left(\sqrt{p_1} - \sqrt{p_2} \right)^2}$ and the union bound, the theorem holds with probability at least $1 - 1/n$. \qed
%
\end{proof}


\textbf{Trade-off between search quality and efficiency:}
By choosing $S$ as Theorem~\ref{thm:main}, we have $\sqrt{\tau_1} - \sqrt{\tau_2} \geq \sqrt{3z(\ln{n})/S}$. 
Assume that the top-$k$ value is $\tau_1$, 
it is clear the more samples we use, the smaller gap between the top-$k$ MIPS values and the other values we can distinguish. 
We note that we can compute and rank $B$ inner products of the top-$B$ points with the largest counter values in the ranking phase.
Increasing $B$ corresponds to increasing the gap $Y - X$ in our analysis, and hence decreasing the failure probability.
This means that both $B$ and $S$ can be used to control such trade-off.

\textbf{Comparison to diamond sampling:} For a fair theoretical comparison, we consider the same setting as in~\cite[Theorem 4]{Diamond} where we want to distinguish  $\bx_{i_1} \cdot \bq \geq \tau$ and $\bx_{i_2} \cdot \bq \leq \tau/4$, and all entries in $\bX$ and $\bq$ are positive\footnote{Diamond sampling's analysis only works on non-negative inputs.}.
Applying Theorem~\ref{thm:main}, wedge sampling needs $S_w \geq 12z\ln{n}/\tau$.
Diamond sampling needs $S_d \geq 12K\|\bq\|_1z\ln{n}/\tau^2$ where all entries in $\bX$ are at most $K$.
Since $K\|\bq\|_1 \geq \tau$ for any $\tau$, wedge requires strictly less samples than diamond.

\textbf{General inputs: } For the general cases, Theorem 1 does not hold anymore.
We observe that there are several ways to convert both data and query into non-negative forms without changing their inner product order.
For instance, by shifting each column $\by_j$ a constant factor (e.g. its minimum or maximum value) dependent on $\sgn(q_j)$, the order of inner products $\bX\bq$ is preserved and therefore Theorem 1 still holds for MIPS in the non-negative transformation space.
We leave this research direction to future work.

\subsection{dWedge: A Simple Deterministic Variant for Budgeted MIPS}\label{sec:implementation}

This subsection presents a significant drawback of wedge sampling to solve top-$k$ MIPS with a $o(n)$ budgeted computation.
We then introduce dWedge, a simple deterministic variant to handle such a drawback. 
For simplicity, we present our approach on non-negative $\bX$ and $\bq$.

\textbf{Drawback:} We observe that wedge sampling first samples the column $j$ and then samples the point $\bx_i$ on this column $j$.
In other words, given a fixed number of samples $S$, wedge sampling allocates $S$ samples to $d$ columns.
Each column $j$ receives $s_j$ samples and $\sum_j s_j = S$.
To ensure that wedge provides an accurate estimate, these $s_j$ samples must approximate the discrete distribution $\by_j/\|\by_j\|_1$ well.
Given the budget of $S = o(n)$ samples, the number of samples $s_j$ of the column $j$ is $S/d \ll n$ in expectation.
Since $s_j \ll n$ and the data set is often dense, it is impossible to approximate the discrete distribution $\by_j/\|\by_j\|_1$ by $s_j$ samples.
In a realistic case of the Netflix data set with $n = 17,770$ and $d = 300$, if we use $S = n$ samples, in expectation we only have nearly~60 samples to approximate a discrete distribution with $17,770$ values.
Hence the performance of wedge (and hence diamond) sampling dramatically degrades in the budgeted setting, as can be seen in the experiment section.

%
%

\begin{algorithm}[!t]
		\KwData {For each dimension $j$, sort data in descending order on $x_{ij}$ and store them in the list $L_j$. Other pre-computed statistics $z$, $c_j$, and the query $\bq$.}
		\KwResult { A counting histogram of sampled data points.}
		 Compute the number of samples $s_j = Sc_jq_j/z$ for each sorted list $L_j$. \\
		\ForEach {sorted list $L_j$}{
		 Select $\bx_i$ in the descending order of $x_{ij}$. \\
		 Increase $\bx_i$'s counter and the current number of samples used of $L_j$ by $\lceil s_j x_{ij}/c_j \rceil$. \\
		 If the current number of samples is larger than $s_j$, stop iterating $L_j$.
	}
		%
		%
		\normalsize
	\caption{dWedge Sampling}
	\label{alg:dWedge} 	
\end{algorithm}

\textbf{dWedge:} 
Observing that wedge sampling carefully distributes $S$ samples to each dimension.
Dimension $j$ receives $s_j = Sc_jq_j/z$ samples and hence the point $\bx_i$ on dimension $j$ will receive $s_jx_{ij}/c_j$ samples in expectation. 
Given $S = o(n)$, we have $s_j \ll n$ and therefore can only sample a few points on the dimension $j$.
Due to this limit, we propose to \emph{greedily} sample $\bx_i$ with the largest $x_{ij}$ values in the column $j$. 
For each selected $\bx_i$, we sample $\lceil s_j x_{ij} /c_j \rceil$ times.

Algorithm~\ref{alg:dWedge} presents our simple heuristic solution on non-negative $\bX$ and $\bq$.
Before executing this algorithm, we need a pre-processing step that sorts all data points in descending order for each dimension.
We iterate these sorted list and greedily sample $\bx_i$ on the descending order.
For each list $L_j$, we stop the iteration when exceeding the number of samples $s_j$. 
 For the general case, dWedge exploits the sign trick as standard wedge, executing on the absolute values of $\bX$ and $\bq$. 
If selected, $\bx_i$'s counter is increased by $\sgn(x_{ij})\sgn(q_j) \lceil s_j x_{ij}/c_j \rceil$.

\textbf{Time complexity:} It is clear that the pre-processing step takes $\BO{dn\log{n}}$ times and $\BO{dn}$ additional space.
dWedge sampling takes $\BO{S}$ time.
If $S \ll n$, we can use a hash table to maintain the counting histogram.
Otherwise, we use a vector of size $n$. 
The running time of dWedge for answering top-$k$ MIPS with a post-processing $B$ inner products is $\BO{S + min(S, n)\log{B} + dB}$.

\textbf{Cost model:} While the cost of extracting top-$B$ points in the ranking phase is larger than the sampling cost in theory, this cost in practice is much smaller due to the optimization of C++ \texttt{std::priority\_queue}.
We observe that dWedge executes two main operations at Step~4--5 for each sample.
Given $S$ samples, dWedge's cost is upper bounded by the cost of computing $2S/d$ inner products.
Hence we can model the operation cost of dWedge for answering top-$k$ MIPS as $2S/d + B$ inner product computation. 
Empirically, this cost model is very accurate due to the simplicity of dWedge.
We will use this simple cost model to estimate the speedup of dWedge and to tune the parameters $S$ and $B$ for comparing to other budgeted MIPS solvers, as shown in the experiment.

\textbf{Relation to Greedy-MIPS~\cite{Yu17}:} This approach exploits the upper bound $\bx_i \cdot \bq \leq d \max_j\{q_jx_{ij}\}$ to construct $B$ candidates.
In principle, it greedily selects $B$ points $\bx_i$ with the largest $q_jx_{ij}$ values for each dimension $j$, then merges them to find top-$B$ candidates with $\max_j\{q_jx_{ij}\}$.
While dWedge shares the same pre-processing step and greedy spirit, there is a significant difference between dWedge and Greedy.
dWedge greedily estimates and differentiates the largest elements on each dimension $j$ using $s_j$ samples.
The more samples used, the more of the largest elements have been considered.
This leads to higher quality of top-$B$ candidates and hence top-$k$ MIPS.
On the other hand, Greedy-MIPS does not have the sampling step to improve the quality of candidates and hence its performance degrades, as shown in the experiment on the Gist data set.

\section{Experiment}\label{sec:experiment}

We implement the sampling schemes and other competitors in C++\footnote{\url{https://github.com/NinhPham/MIPS}} using -O3 optimization and conduct experiments on a 2.80 GHz core i5-8400 32GB of RAM.
We present empirical evaluations to verify our claims, including: (1) Wedge is competitive (often superior) to Diamond; (2) dWedge provides better accuracy than both Wedge and Greedy-MIPS; and (3) dWedge is more flexible and often superior to competitive LSH-based solvers~\cite{Neyshabur15,Yan18} for the budgeted MIPS.

For measuring the accuracy and efficiency, we used the standard Precision@10 and the speedup over the brute-force algorithm, defined as follows.
\begin{align*}
\textrm{Precision@10} &= |\textrm{Retrieved top-10 } \cap \textrm{ True top-10}|/10 \, , \\
\textrm{Speedup} &= \textrm{Running time of brute-force }/\textrm{ Running time of algorithm.}
\end{align*}
%

\subsection{Experiment Setup and Data Sets}

Since diamond exploits wedge, applying dWedge to diamond derives a new variant, called dDiamond\footnote{This variant is not deterministic due to the randomness from the basic sampling.}.
The list of all implemented algorithms includes: (1) The traditional wedge (\emph{Wedge}) and diamond (\emph{Diamond}) sampling and the proposed solutions \emph{dWedge} and \emph{dDiamond}; (2) The Greedy-MIPS approach (\emph{Greedy})~\cite{Yu17}; (3) Representative LSH-based solutions, including \emph{SimpleLSH}~\cite{Neyshabur15} and the recently improved algorithm \emph{RangeLSH}~\cite{Yan18}; and (4) Brute-force algorithm with the Eigen-3.3.4 library\footnote{\url{http://eigen.tuxfamily.org/index.php?title=Main_Page}} for the extremely fast C++ matrix-vector multiplication. 

We conduct experiments on standard real-world data sets\footnote{\url{ https://drive.google.com/drive/folders/1BHpiaii6Ur0rKSy5c9AFVwAhfLUQMXFE}}, including Netflix, Yahoo, and Gist.
For the sake of comparison, we use the Netflix-200 ($n = 17,770; d = 200$) from~\cite{Yu17}, Netflix-300 ($n = 17,770; d = 300$) and Yahoo ($n = 624,961; d = 300$) from~\cite{Cremonesi10}, and Gist ($n = 1,000,000; d = 960$).
For Netflix and Yahoo, the item matrices are used as the data points.
We randomly pick 1000 users from the user matrices to form the query sets.
All randomized results are the average of 5 runs of the algorithms.

\subsection{Comparison between Wedge and Diamond Schemes}

\begin{figure} [!t]
	\centering
	\includegraphics[width=\textwidth]{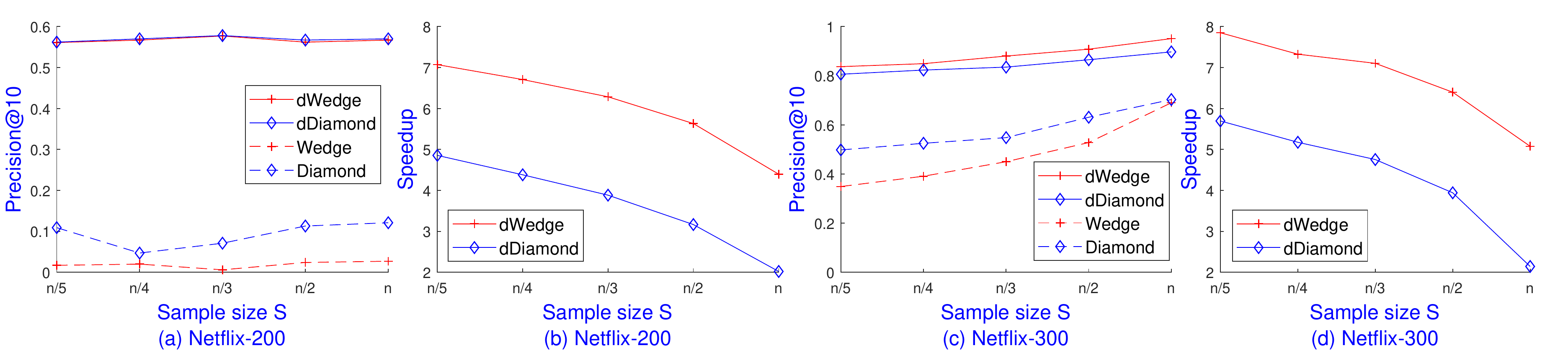}
	\caption{Comparison of accuracy and speedup between dWedge, dDiamond, Wedge and Diamond on Netflix-200 and Netflix-300 when fixing $B = 100$ and varying $S$.}
	\label{fig:wedge_diamond}
\end{figure}

In this subsection, we compare the top-$k$ MIPS performance on Netflix between wedge-based and diamond-based schemes using deterministic and randomized generators.
For each sampling scheme, we consider two corresponding variants, including Wedge, Diamond, dWedge and dDiamond.
We measure their performance on Precision@10 value and speedup over the brute-force search where we varied the sample size $S$ and fix $B = 100$ for post-processing.

Figure~\ref{fig:wedge_diamond} reveals that the proposed deterministic generator gives higher accuracy than the randomized one.
dWedge and dDiamond outperform Wedge and Diamond, respectively, over a wide range of $S$.
Especially, on Netflix-200, Wedge and Diamond suffer very low accuracy whereas dWedge and dDiamond return more than 55\%.
Netflix-300 shows the superiority of both dWedge and dDiamond with the accuracy at least 80\%.
We note that the two Netflix versions are generated by different matrix factorization tools.
On Netflix-300, the true top-$k$ points tend to dominate the others on each dimension, hence dWedge and dDiamond achieve higher accuracy than on Netflix-200.

In term of speedup, dWedge runs significantly faster than dDiamond, especially at least twice faster when $S = n$.
Wedge variants run faster than diamond variants since diamond requires the basic sampling step which requires expensive cost for random accesses.
This gap will be more substantial on large data sets but not reported here due to the lack of space.

\subsection{Comparison between dWedge and Greedy-MIPS}

In this subsection, we compare the top-$k$ MIPS performance between dWedge and Greedy on all used data sets.
Since the cost of the screening phase of Greedy is implicitly governed by the number of inner product computation $B_g$ in the ranking phase, the larger $B_g$ provides the higher accuracy for Greedy.
Note that dWedge's cost is about $2S/d + B$ inner product computation. 
Since we want to show that dWedge always runs as fast as Greedy but achieves higher accuracy, we often set $B_g > 2S/d + B$ inner product computation for Greedy.
In particular, for small Netflix-200 and Netflix-300, we use $B_g = 2S/d + B + 50$ and $B_g = 2S/d + B + 20$, respectively.
For the large Yahoo data set, we use $B_g = B$ due to the significant cost of the screening phase of Greedy.

\begin{figure} [!t]
	\centering
	\includegraphics[width=\textwidth]{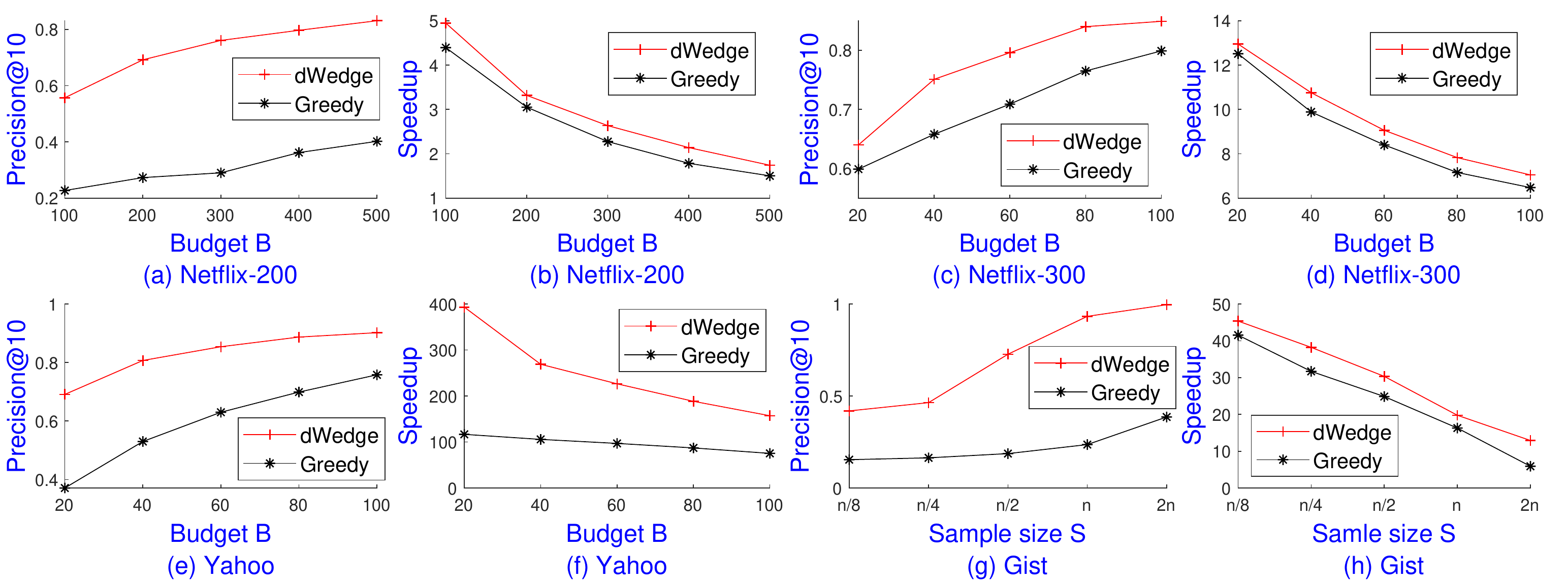}
	\caption{Comparison of accuracy and speedup between dWedge and Greedy when fixing $S=10,000$ and varying $B$ on Netflix-200 (a, b); fixing $S = 4,500$ and varying $B$ on Netflix-300 (c, d); fixing $S = 4,500$ and varying $B$ on Yahoo (e, f); and fixing $B = 200$ and varying $S$ on Gist (g, h).}
	\label{fig:wedge_greedy}
\end{figure}

Figure~\ref{fig:wedge_greedy} (a--d) shows that, given our setting above, dWedge runs slightly faster and provides dramatically higher accuracy than Greedy on the two Netflix versions.
dWedge achieves up to 80\% Precision@10 on both versions while Greedy bears at most 40\% on Netflix-200.
On Yahoo in Figure~\ref{fig:wedge_greedy} (e, f), since the screening phase of Greedy incurs remarkable cost, dWedge yields substantially higher accuracy and speedup.
We note that the Eigen's batch inner product computation runs nearly 20 times faster than a standard implementation on our machine.
Hence, using our cost model, dWedge's speedup can be estimated as $n/(40S/d + 20B)$.

In order to show the benefit of the sampling phase of dWedge, we measure its performance on Gist while fixing $B = 200$ and varying $S$.
Again for Greedy, we add additional inner product computation into the post-processing phase so that dWedge and Greedy achieve the similar speedup, as can be seen in Figure~\ref{fig:wedge_greedy} (h).
Figure~\ref{fig:wedge_greedy} (g) shows that Greedy suffers from very small accuracy with at most 40\% while dWedge achieves nearly perfect recall, i.e. 99\% when $S = 2n$.
In general, dWedge with two parameters, i.e. number of samples $S$ and number of post-processing inner products $B$, governs the trade-off between search efficiency and quality more efficiently than Greedy.
Particularly, dWedge offers a least 80\% Precision@10 on all 4 data sets with remarkable speedups while Greedy does not.

\subsection{Comparison between dWedge and LSH-based Schemes}

This subsection shows experiments comparing the performance between dWedge and LSH-based schemes, including SimpleLSH~\cite{Neyshabur15} and RangeLSH~\cite{Yan18} on Yahoo and Gist.
There are two strategies of using LSH for top-$k$ MIPS: constructing binary codes for efficient estimation and constructing hash tables for efficient lookups.
We consider both strategies in our experiments.

SimpleLSH~\cite{Neyshabur15} transformed data $\bx \mapsto \{\bx/m, \sqrt{1 - \|\bx\|_2/m} \}$ and query $\bq \mapsto \{\bq/\|\bq\|_2, 0\}$ where $m$ is the maximum 2-norm value of all data points.
RangeLSH~\cite{Yan18} partitions the data set into several partitions and applies SimpleLSH on each partition.
We show empirically that the MIPS values in the transformation space are much smaller and therefore degrades the efficiency of both LSH schemes. 


\textbf{LSH estimation: } For the first strategy, the LSH time complexity is dominated by the hash evaluation, i.e. $\BO{dh}$, and the linear cost of Hamming distance computation, i.e. $\BO{nh}$, where $h$ is the code length.
In practice, we can use the Eigen library with fast matrix-vector multiplication to speed up the hash computation.
Furthermore, the Hamming distance computation is also fast due to the \texttt{\_builtin\_popcount} function of GCC compilers.
Hence, LSH-based estimation is still much faster than brute-force search.
To compare dWedge with LSH schemes, we set $S = 6,000 \approx n/100$ for Yahoo and $S = 2n$ for Gist.
Both dWedge and LSH-based methods use $B = 100$ inner products for post-processing.
We evaluate SimpleLSH and RangeLSH over a wide range of code length $h$. 

\begin{figure*} [!t]
	\centering
	\includegraphics[width=1.0\textwidth]{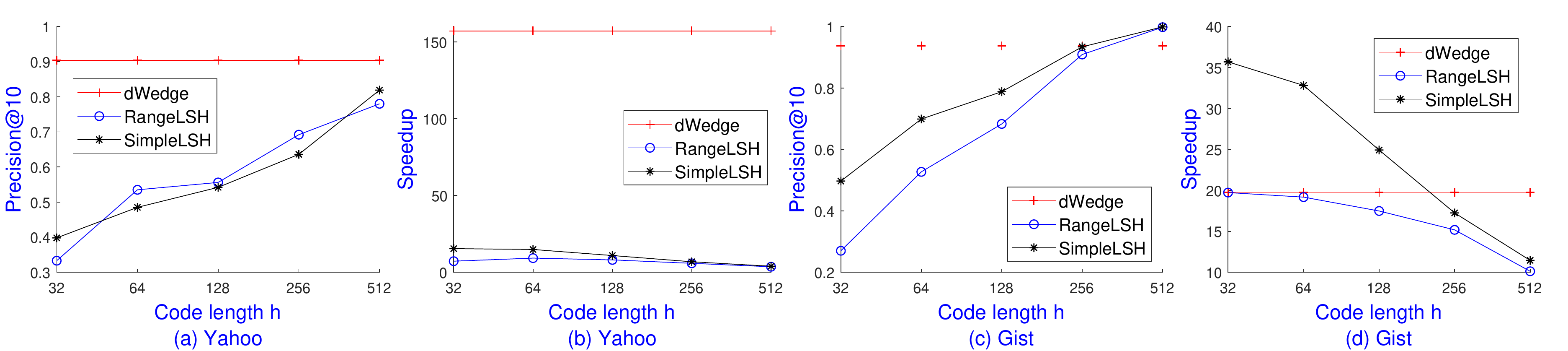}
	\caption{Comparison of accuracy and speedup between dWedge ($S = n/100$ for Yahoo, $S = 2n$ for Gist), SimpleLSH and RangeLSH when fixing $B = 100$ and varying $h$.}
	\label{fig:wedge_lshcode}
\end{figure*}

Figure~\ref{fig:wedge_lshcode} shows that dWedge is superior to LSH schemes on these two data sets.
On Yahoo, dWedge achieves 90\% Precision@10 with approximate 180x speedup whereas the LSH schemes give at most 80\% with negligible speedup when $h = 512$.
On Gist, LSH schemes only achieve slightly higher accuracy but run twice slower than dWedge when $h = 512$.
After transformation, both the MIPS values and the gap between ``close'' and ``far apart'' points are often very small.
That requires significantly large code length, i.e. $h = 512$, to achieve acceptable accuracy.
For the other values of $h$, dWedge always returns higher accuracy and runs dramatically faster than RangeLSH.
For the same speedup setting, dWedge still provides higher accuracy than SimpleLSH but not reported here due to the lack of space.
We note that dWedge provides significant speedup on Yahoo compared to Gist because of the much smaller number of samples $S$.

In our experiment, SimpleLSH is competitive with RangeLSH.
RangeLSH roughly estimates $\bx \cdot \bq$ as $m_i \cos{\pi (1 - p)}$ where $p$ is the collision probability and $m_i$ is the maximum 2-norm of the partition of $\bx$.
While $m_i$ is always non-negative, $\cos{\pi (1 - p)}$ might be negative.
A small error induced from LSH functions will change the sign of $\cos{\pi (1 - p)}$ and therefore change the results completely.
On the other hand, SimpleLSH does not suffer this issue.
Since the comparison of LSH schemes is out of the scope, we do not discuss it in more detail.

\textbf{LSH lookups: } For the second strategy, since we do not know the collision probability between the query and all data points, it is impossible to tune LSH parameters for any budget cost.
Given a fixed number of hash tables $L$, the number of concatenating hash functions $h$ will control the number of collisions and therefore the candidate set size. 
In particular, large values of $h$ will decrease the collision probability and hence candidate set size and vice versa.
The query complexity will be dominated by $\BO{dhL + dB}$ where the first term is from the hash computation and the second term comes from the post-processing phase of computing $B$ inner products.
Again, we used the Eigen library for speeding up the hash computation.

\begin{table*}[!t]
		\caption{Comparison of accuracy and running time of the screening  and ranking phases between dWedge and SimpleLSH on Yahoo when fixing $B = 40$ and varying $h$.}
	\centering
	\begin{tabular}{|c|c|c|c|c|c|} 
		\hline
		Method 
		& \begin{tabular}{@{}c@{}}dWedge \\ $S$ = 6,000\end{tabular} 
		& \begin{tabular}{@{}c@{}}SimpleLSH \\ $h$ = 8\end{tabular} 
		& \begin{tabular}{@{}c@{}}SimpleLSH \\ $h$ = 10\end{tabular}
		& \begin{tabular}{@{}c@{}}SimpleLSH \\ $h$ = 12\end{tabular} 
		& \begin{tabular}{@{}c@{}}SimpleLSH \\ $h$ = 16\end{tabular}\\ 
		\hline
		Screening time (ms) & \textbf{0.158} & 0.342 & 0.444 & 0.649 & 0.998\\ 
		\hline
		Ranking time (ms) & \textbf{0.144} & 0.255 & 0.224 & 0.242 & 0.043\\ 
		\hline
		Total time (ms) & \textbf{0.312} & 0.646 & 0.731 & 0.984 & 1.171\\ 
		\hline
		Precision@10 & \textbf{82.4\%} & 20.4\% & 21.6\% & 21.8\% & 5.3\% \\
		\hline
	\end{tabular}
	\label{tab:wedge_lshtable}
\end{table*}

Since we need to set up these parameters before querying, LSH provides a fixed trade-off between search efficiency and quality.
To compare with dWedge, we first fix $L = 512$ as suggested in previous LSH-based MIPS solvers~\cite{Neyshabur15,Shrivastava14} and vary the parameter $h$.
Due to the similar results\footnote{RangeLSH provides slightly better than SimpleLSH with small $h$ values. When $h \geq 20$, they will share the same results of very low accuracy on both data sets.}, we report the representative results of SimpleLSH on Yahoo.
As explained in Section 4.3, the screening cost of dWedge with $S$ samples is similar to the cost of computing $2S/d$ inner products.
Hence, dWedge with $S = 6,000 \approx n/100$ and $B = 40$ balances the screening time and ranking time on Yahoo with $d = 300$.
We also set $B = 40$ as the candidate size of SimpleLSH for the sake of comparison.

Table~\ref{tab:wedge_lshtable} shows the comparison between dWedge and SimpleLSH in several components, including the screening time, ranking time, total running time, and accuracy.
The screening time of dWedge and SimpleLSH are the sampling and hash evaluation time, respectively.
It is clear that the screening time and ranking time dominate the total running time.
Furthermore, dWedge has not only similar screening and ranking time but also negligible overheads, which expresses the simplicity of dWedge on tuning parameters.
Given such a balanced setting, dWedge not only runs at least twice faster but also achieves 82.4\% accuracy, which is significantly higher than the maximum 21.8\% of SimpleLSH with $h = 12$.

Regarding the screening time, dWedge runs at least twice faster than SimpleLSH with $h = 8$.
Regarding the ranking time, dWedge is also faster than SimpleLSH, except for the case $h = 16$, though we use the same $B = 40$.
dWedge chooses the top-$B$ with the largest counter values hence has a higher chance to retrieve top-$k$ MIPS earlier. 
This significantly reduces the number of insertions into the priority queue.
On the other hand, such a chance of SimpleLSH is uniform, hence SimpleLSH has no such benefit.

For SimpleLSH, the screening time increases due to the increase of the hash evaluation time.
However, the ranking time of SimpleLSH is similar when $h < 12$ and then drops to a negligible amount at $h = 16$.
Similarly, Precision@10 slightly increases when $h$ increases and then significantly drops at $h = 16$.
This is due to the fact that SimpleLSH has enough candidates (i.e. $B = 40$) when $h \leq 12$.
However, it does not have enough collisions on all $L$ hash tables when $h = 16$.
We observe that the average value of $B$ is just 6 for all queries.
This fact is not surprising. 
A simple computation indicates that the average maximum MIPS is 0.24 after transformation.
Hence, the number of collisions is at most $n(1 - (1 - 0.24^h)^L) \approx 0$ when $h = 16$ and $L = 512$.
We note that this observation is also true with RangeLSH.
Due to the lack of space, we will not discuss RangeLSH's performance here.
In general, LSH with both strategies is inferior to dWedge on our benchmark data sets. 
LSH schemes cannot govern effectively the trade-off between search quality and efficiency as dWedge.

\section{Conclusions}\label{sec:conclusion}

This paper studies top-$k$ MIPS given a limited budget of computational operations and investigates recent advanced sampling algorithms, including wedge and diamond sampling to solve it.
We theoretically and empirically show that wedge sampling is competitive (often superior) to diamond sampling on approximating top-$k$ MIPS regarding both efficiency and accuracy.

We propose dWedge, a very simple deterministic wedge variant for budgeted top-$k$ MIPS.
dWedge runs significantly faster or provides higher accuracy than other competitive budgeted MIPS solvers while maintaining an accuracy of at least 80\% on our real-world benchmark data sets.




%



%
%

\end{document}